%% file: main.tex
\documentclass[conference]{IEEEtran}
\usepackage{amsmath}
\usepackage{bm}
\usepackage{amssymb}
\usepackage{graphicx}
\usepackage{color, colortbl} 
\usepackage[dvipsnames,svgnames,x11names]{xcolor}
\makeatother
\usepackage{algpseudocode,algorithm}
\usepackage{hyperref}
\usepackage{enumerate}
\usepackage[caption=false,font=footnotesize]{subfig}
\usepackage{cite}
\usepackage{url,comment}
\usepackage[normalem]{ulem}
\usepackage[shortlabels]{enumitem}
\usepackage{amsthm}
\usepackage{pgfplots}
\usepackage{tikz}
\usepackage{mathtools}
\usepackage{balance}
\usepackage{acronym}
\usepackage{amssymb}
\usepackage[dvipsnames]{xcolor}
\usepackage{cite}
\usetikzlibrary{calc}
\makeatletter
\newcommand{\gettikzxy}[3]{%
  \tikz@scan@one@point\pgfutil@firstofone#1\relax
  \edef#2{\the\pgf@x}%
  \edef#3{\the\pgf@y}%
}
\usepackage{amsmath,amssymb,amsfonts,bm}
\usepackage{xcolor}
\usepackage{relsize}
\usepackage{amsmath,amsthm,amssymb,amsfonts,bm}
\usepackage{comment}
\usepackage{amsmath}
\usepackage{acronym}
\newtheorem{proposition}{Proposition}
\newtheorem{rem}{Remark}
\acrodef{bs}[BS]{base station}
\acrodef{ue}[UE]{user equipment}
\acrodef{los}[LOS]{line-of-sight}
\acrodef{aoa}[AoA]{angle-of-arrival}
\acrodef{aod}[AoD]{angle-of-departure}
\acrodef{toa}[ToA]{time-of-arrival}
\acrodef{tdoa}[TDoA]{time-difference-of-arrival}
\acrodef{ris}[RIS]{reconfigurable intelligent surface}
\acrodefplural{ris}[RISs]{reconfigurable intelligent surfaces}
\acrodef{tx}[Tx]{transmitter}
\acrodef{rx}[Rx]{receiver}
\acrodefplural{rx}[Rxs]{receivers}
\acrodef{crb}[CRB]{Cram\'er-Rao lower bounds}
\acrodef{rss}[RSS]{received signal strength}
\acrodef{los}[LOS]{line-of-sight}
\acrodef{nlos}[NLOS]{non line-of-sight}
\acrodef{dft}[DFT]{discrete Fourier transform}
\acrodef{fft}[FFT]{fast Fourier transform}
\acrodef{fim}[FIM]{Fisher information matrix}
\acrodef{upa}[UPA]{uniform planar array}
\acrodefplural{upa}[UPAs]{uniform planar arrays}
\acrodef{peb}[PEB]{position error bound}
\acrodef{snr}[SNR]{signal-to-noise ratio}
\acrodef{sre}[SRE]{smart radio environment}
\acrodefplural{sre}[SRE]{smart radio environments}
\acrodef{mimo}[MIMO]{multiple-input  multiple-output}
\acrodef{rfid}[RFID]{radio-frequency identification}
\acrodef{siso}[SISO]{single-input single-output}
\acrodef{miso}[MISO]{multiple-input single-output}
\acrodef{ici}[ICI]{inter-carrier interference}
\acrodef{iid}[iid]{independent and identically distributed}
\acrodef{ml}[ML]{maximum likelihood}
\acrodef{cdf}[CDF]{cumulative distribution function}
\acrodef{ofdm}[OFDM]{orthogonal frequency-division multiplexing}
\acrodef{qos}[QoS]{Quality of Service}
\acrodef{ula}[ULA]{uniform linear array}

\newcommand{\Es}{E_{\mathrm{s}}}

\newcommand{\George}[1]{\textcolor{red}{\footnotesize{\textsf {[George: #1]}}}}

\bstctlcite{IEEEexample:BSTcontrol}
\begin{document}
\title{RIS-Enabled Self-Localization: Leveraging Controllable Reflections With Zero Access Points}

\author{\IEEEauthorblockN{Kamran Keykhosravi\IEEEauthorrefmark{1}, Gonzalo Seco-Granados\IEEEauthorrefmark{2},  George C. Alexandropoulos\IEEEauthorrefmark{3}, and Henk Wymeersch\IEEEauthorrefmark{1}}\\
\IEEEauthorblockA{\IEEEauthorrefmark{1} Department of Electrical Engineering, Chalmers University of Technology, Sweden \\
\IEEEauthorrefmark{2}  Department of Telecommunications and
Systems Engineering, Universitat Auton\`{o}ma de Barcelona, Spain \\
\IEEEauthorrefmark{3}Department of Informatics and Telecommunications,
National and Kapodistrian University of Athens, Greece\\
emails: \{kamrank, henkw\}@chalmers.se, gonzalo.seco@uab.es, alexandg@di.uoa.gr}}
\maketitle

\begin{abstract}
Reconfigurable intelligent surfaces (RISs) are one of the most promising technological enablers of the next (6th) generation of wireless systems. In this paper, we introduce a novel use-case of the RIS technology in radio localization, which is enabling the user to estimate its own position via  transmitting orthogonal frequency-division multiplexing  (OFDM) pilots and processing the signal reflected from the RIS. We demonstrate that user localization in this scenario is possible by deriving  Cram\'er-Rao lower bounds  on the positioning error and devising a low-complexity position estimation algorithm. We consider random and directional RIS phase profiles and apply a specific temporal coding to them, such that the reflected signal from the RIS can be separated from the uncontrolled multipath.
Finally, we assess the performance of our position estimator for an example system, and show that the proposed algorithm can attain the derived bound at high signal-to-noise ratio values. 
\end{abstract}

\IEEEpeerreviewmaketitle

\section{Introduction}
Reconfigurable intelligent surfaces (RISs) are expected to revolutionize  wireless systems by enabling smart radio environments, where the propagation channel can also be programmed to improve the \ac{qos} \cite{huang2019reconfigurable,wu2019towards}. 
In general, an RIS can be modeled as a planar array of sub-wavelength unit cells, each of which can scatter the impinging signal after modulating its phase in a controlled fashion \cite{alexandg_2021}. The RIS phase profile can then be designed for optimal beamforming, localization, or interference management. Recently, a large body of research has been devoted to study the modeling, control, phase-profile design, and potential use-cases of RISs for both communications and radio localization \cite{rise6g}.
With radio localization (which is the main topic of this paper) RISs can enable or boost the accuracy of \ac{ue} positioning by providing: \emph{i)}  a strong reflected signal path towards the \ac{ue}, and \emph{ii)} a reference position \cite{henk_radio}. Being cost- and energy-efficient, RISs have the potential to boost/enable radio localization in a wide variety of situations, specifically in those where GPS signal is unavailable or weak, e.g., in city canyons, indoor environments, and tunnels. 

Many studies have been conducted to investigate RIS-aided localization in different scenarios, which can be categorized in terms of the operating regime (near-field \cite{dardari_spawk,sha_18,abu2021near,rahal2021ris}, far-field \cite{fascista2021ris,keykhosravi2020siso,keykhosravi2021semi}), RIS placement (at the \ac{bs} side \cite{sha_18,abu2021near}, at the UE side \cite{keykhosravi2021semi}, or as a separate reflector \cite{habo_rss, elzanaty_TSP}), wireless settings (\ac{mimo}\cite{elzanaty_TSP,Yiming_ICC21}, \ac{miso} \cite{fascista2021ris}, and \ac{siso} \cite{keykhosravi2020siso,rahal2021ris}), etc. Furthermore, RISs can be used to assist the radar systems to improve the target detection capabilities \cite{radarStefano,lu2021intelligent}. In \cite{sha_18}, \ac{crb} on the localization error were established for the near-field of a continuous RIS and the effects of the RIS size were studied. For the same scenario, the effects of limited RIS phase resolution was investigated in \cite{cramer_juan}. In \cite{abu2021near}, localization with an RIS acting as a lens has been considered and an estimation algorithm was proposed. Near-field localization with one BS and an RIS has been studied in \cite{dardari_spawk , rahal2021ris}, where it is shown that the wavefront curvature can be used to localize the UE, even if the direct path from the BS to \ac{ue} is blocked. %
Specifically, the authors in \cite{dardari_spawk} illustrated that by using a large stripe-like RIS, positioning can be performed even if the RIS is severely obstructed. For a generic \ac{mimo} setting equipped with an RIS, the \acp{crb} have been developed on  location and orientation errors  in \cite{elzanaty_TSP}. In \cite{Yiming_ICC21}, the authors considered an RIS-aided \ac{mimo} scenario, derived the \acp{crb}, and used the bounds to optimize the RIS phase profile for localization.
Furthermore, it has been shown that localization and \ac{ue} synchronization can be performed in a \ac{miso} \cite{fascista2021ris} and even in a \ac{siso} \cite{keykhosravi2020siso} setup with the help of a single RIS, when far-field conditions hold. 

\begin{figure}
    \centering
    \begin{tikzpicture}
    \node (image) [anchor=south west]{\includegraphics[width=5cm]{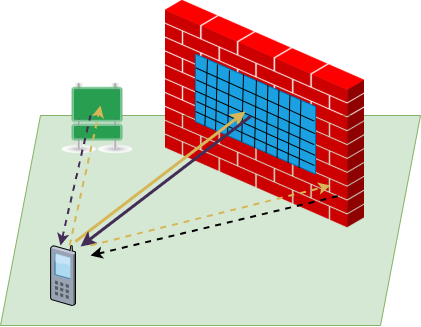}};
    \gettikzxy{(image.north east)}{\ix}{\iy};
    \node at (.12*\ix,.3*\iy){\footnotesize $\bm{p}_{\mathrm{u}}$};
    \node at (.1*\ix,.07*\iy){\footnotesize UE};
    \node at (.25*\ix,.78*\iy){\footnotesize Scatterer};
    \node at (.6*\ix,.65*\iy){\footnotesize \color{white} $\bm{p}_{\mathrm{r}}$};
    \node[rotate=-20,fill=red] at (.6*\ix,.8*\iy){\footnotesize \color{white} RIS};
    \end{tikzpicture}
    \caption{The considered system setup with a single-antenna UE and an RIS. The wireless environment is assumed to also include passive objects contributing to multipath signal reflection.}
    \label{fig:setup}
\end{figure}
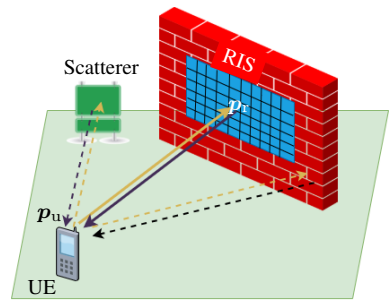

In this paper, we present a novel use-case for RIS-enabled localization, where an RIS is used to reflect the signal transmitted from the \ac{ue} back to the \ac{ue} itself, and there are no access points or BSs present (see Fig.\,\ref{fig:setup}). We establish the channel model for \ac{ofdm} transmission by taking into account the undesired multipath from the surrounding environment. This multipath is then removed via a specific RIS profile design and a postprocessing step at the \ac{ue} side. Next, we estimate the user position by first obtaining a coarse estimate and then refining it to maximize the likelihood function. Finally, we evaluate our estimation method and compare its accuracy to our derived \acp{crb}. It is shown that the proposed estimator attains the bounds at high \acp{snr}.  

\paragraph*{Notations} Vectors are shown by bold lower-case letters and matrices by bold upper-case ones. We use $[\bm{a}]_i$ to indicate the $i$th element of the vector $\bm{a}$ and $[\bm{A}]_{i,j}$ to denote the element at the $i$th row and the $j$th column of matrix $\bm{A}$. Furthermore, the subindex $i:j$ is used to specify all the elements between $i$ and $j$. All vectors are column vectors by default. Transpose, Hermitian, and complex-conjugate operations are denoted by $(\cdot)^{\top}$, $(\cdot)^{\mathrm{H}}$, and $(\cdot)^{*}$, respectively. The Hadamard product is indicated by $\odot$.

\section{System and channel model}
We consider a  single-antenna full-duplex UE with unknown location $\bm{p}_{\mathrm{u}} \in \mathbb{R}^3$ and a single RIS with known center location $\bm{p}_{\mathrm{r}}\in \mathbb{R}^3$ and orientation $\bm{R}\in \text{SO}(3)$. We assume that the RIS is handled by a dedicated controller with whom the UE can communicate for establishing synchronization \cite{rise6g}. The UE location relative to the RIS is defined as $\bm{p}_{\mathrm{ur}}=\bm{p}_{\mathrm{u}}-\bm{p}_{\mathrm{r}}$. The RIS has $M$ elements in a square configuration. We indicate the location of the $m$th RIS element by $\bm{p}_{\mathrm{r},m}$. The UE transmits $T$ \ac{ofdm} signals with $N$ subcarriers and receives the backscattered signals from the RIS and from some other scatterers in the environment. We assume that the UE-RIS (or RIS-UE) channel has negligible \ac{nlos} components. Furthermore, we assume that all the transmitted pilot symbols are equal to $\sqrt{\Es}$, with $\Es$ being the symbol energy. Then, the received signal at the UE at the discrete time $t=1,2,\dots, T$ can be modeled as follows:
\begin{align}\label{eq:yt}
    \bm{y}_t&\triangleq\sqrt{\Es}\beta_0 \bm{d}(\tau_0) \bm{a}(\bm{p}_{\mathrm{ur}})^\top\bm{\Omega}_t\bm{a}(\bm{p}_{\mathrm{ur}})\nonumber\\
    & \quad \quad + \sqrt{\Es}\sum_{l=1}^L \beta_l \bm{d}(\tau_l) +\bm{n}_t.
\end{align}
Here, $\tau_0\triangleq2\Vert \bm{p}_{\mathrm{ur}} \Vert/c$ is the transmission delay, where $c$ is the speed of light. The delay steering vector is defined as
\begin{align}
\bm{d}(\tau) \triangleq [1, e^{-\jmath 2\pi \tau \Delta f }, \dots, e^{-\jmath 2\pi \tau (N-1)\Delta f }]^{\top},    
\end{align}
where $\Delta f$ is the subcarrier spacing. The complex channel gain is denoted by $\beta_0\in\mathbb{C}$.
The vector $\bm{a}(\bm{p}_{\mathrm{ur}})\in \mathbb{C}^{M}$ indicates the RIS response, having the elements for $m=1,2,\ldots,M$:
\begin{align}\label{eq:aVec}
    [\bm{a}(\bm{p}_{\mathrm{ur}})]_m=\exp \left(\jmath \frac{2\pi}{\lambda} (  \Vert\bm{p}_{\mathrm{u}}-\bm{p}_{r}\Vert-\Vert\bm{p}_{\mathrm{u}}-\bm{p}_{\mathrm{r},m}\Vert) \right), 
\end{align}
 where $\lambda \triangleq c/f_{\mathrm{c}}$ is the wavelength with $f_{\mathrm{c}}$ being the carrier frequency. Furthermore, we have $\bm{\Omega}_t \triangleq \mathrm{diag}(\bm{\omega}_t)$, where $\bm{\omega}_t \in \mathbb{C}^{M}$ is the RIS phase profile, i.e., for each $m$, the amplitude of $[\bm{\omega}_t]_m$ is one and its phase indicates the phase shift applied by the $m$th element of RIS to the impinging signal. The noise term $\bm{n}_t$ is assumed to be spatially and temporally white with covariance $\sigma_{\mathrm{n}}^2\bm{I}_N$. The number of uncontrolled multipath components is denoted by $L$. The path loss and the delay for the $\ell$th path are denoted by $\beta_{\ell}$ and $\tau_{\ell}$, respectively.  By using the following definition: 
 \begin{align}\label{eq:bDef}
     \bm{b}(\bm{p}_{\mathrm{ur}}) \triangleq \bm{a}(\bm{p}_{\mathrm{ur}}) \odot \bm{a}(\bm{p}_{\mathrm{ur}}),
 \end{align}
 we can rewrite \eqref{eq:yt} in the more compact form:
 \begin{align}\label{eq:ytCompact}
    \bm{y}_t&=\sqrt{\Es}\beta_0 \bm{d}(\tau_0) \bm{b}(\bm{p}_{\mathrm{ur}})^{\top} \bm{\omega}_t + \sqrt{\Es}\sum_{l=1}^L \beta_l \bm{d}(\tau_l) +\bm{n}_t.
\end{align}
We assume that the UE has the knowledge of the RIS location and orientation and also the RIS phase profiles $\bm{\Omega}_t$, which change over time; this information can be communicated by the RIS controller. This controller also coordinates the start of the pilot transmission with the UE. Using these system control settings, the UE's goal is to estimate its position via its received signal $\bm{y}_t$. We next establish the requirement for the RIS element spacing in order to avoid ambiguities in the UE position estimation.

\begin{rem}\label{remark:lambda}
To avoid grating lobes, which causes ambiguity in the position estimation at the far field, the RIS element spacing should be at least $\lambda/4$. To show this, we rewrite the far-field approximation of the vector $\bm{b}(\bm{p}_{\mathrm{u}})$ in \eqref{eq:bDef} for a two-dimensional (2D) case with an $1\times M$ \ac{ula} RIS, as $[\bm{b}_{\mathrm{ff}}(\theta)]_m=\exp (\jmath \frac{4\pi}{\lambda}  m d \sin{\theta} ),$ where $\theta$ is the angle between $\bm{p}_{\mathrm{ur}}$ and the RIS normal. Since we assume that the UE is placed in the front side of the RIS, we have that $\theta=[-\pi/2, \pi/2]$, and hence, $\sin(\theta)=[-1, 1]$. Therefore, the phase difference between two consecutive entries of the vector $\bm{b}_{\mathrm{ff}}(\theta)$ is between $[-\jmath 4\pi d/\lambda \ \jmath 4\pi  d/\lambda]$. Consequently, it must hold $8\pi d/\lambda\leq 2\pi$, resulting in $d\leq\lambda/4$, to be able to find $\theta$ without ambiguity.
\end{rem}

\section{RIS phase profile design}\label{sec:risPhaseDesign}
In this section, we study the design of the RIS phase profile $\bm{\omega}_t$. First, we describe how to design $\bm{\omega}_t$ to remove the uncontrolled multipath at the UE's receiver, and then, we introduce the random and directional RIS phase profiles, which are used in simulation results in Section~\ref{sec:simulationResults}.

\subsection{Multipath Removal}\label{sec:removingMultipath}
To remove the multipath from the received signal in \eqref{eq:ytCompact},  using the method presented in \cite{keykhosraviMulti}. First, we design the RIS phase profiles for half of the transmission time and denote them by $\tilde{\bm{\omega}}_{\tilde{t}} \in \mathbb{C}^{M}$, where $\tilde{t}=1,2,\dots, T/2$ ($T$ is assumed even). The vectors $\tilde{\bm{\omega}}_{\tilde{t}}$ can be chosen based on a random codebook (explained later on in Section\,\ref{sec:randomCodeBook}), a directional one (explained in Section\,\ref{sec:directionalCodeBook}), or any other strategy. Then, we let $\bm{\omega}_{2\tilde{t}-1} =  \tilde{\bm{\omega}}_{\tilde{t}}$ and $\bm{\omega}_{2\tilde{t}} = -  \tilde{\bm{\omega}}_{\tilde{t}}$. Consequently, at the UE receiver, we compute the following quantity:
\begin{align}
\tilde{\bm{y}}_{\tilde{t}} &= \frac{1}{2}(\bm{y}_{2\tilde{t}-1}-\bm{y}_{2\tilde{t}})\\
&=\sqrt{\Es}\beta_0 \bm{d}(\tau_0) \bm{b}(\bm{p}_{\mathrm{ur}})^{\top} \tilde{\bm{\omega}}_{\tilde{t}}  +\tilde{\bm{n}}_t,\label{eq:yTilde}
\end{align}
where \eqref{eq:yTilde} follows from \eqref{eq:ytCompact} and $\tilde{\bm{n}}_t$ is a white complex Gaussian noise with covariance matrix $\sigma_{\mathrm{n}}^2/2\bm{I}_N$. It can be seen that the multipath contribution has been disappeared in \eqref{eq:yTilde}.  

\subsection{Random Codebook}\label{sec:randomCodeBook}
The random codebook is used when there exists no prior knowledge of the UE location.
With this codebook, we have $[\tilde{\bm{\omega}}_t]_m=\exp(\jmath  \theta_{m,t})$, where $\theta_{m,t}$ $\forall$$m,t$ are chosen randomly according to the uniform distribution over the interval $[0,2\pi)$, each one independently from each other.

\subsection{Directional Codebook}\label{sec:directionalCodeBook}
We use the directional codebook when there exists a prior knowledge about the UE position. We assume that an approximate location $\bm{q}_{\mathrm{u}}$ is available, which is distributed uniformly within the sphere $\mathcal{S}(\bm{p}_{\mathrm{u}},\delta)$, with $\delta$ being the uncertainty radius, and for any $\bm{c}\in\mathbb{R}^3$ and $r\in \mathbb{R}^+$ holds:
\begin{align}
    \mathcal{S}(\bm{c},r) = \{\bm{x}\in\mathbb{R}^3 \big\vert  \Vert\bm{c}-\bm{x} \Vert \leq r\}.
\end{align}

To generate a directional codebook, we first select $T/2$ points $\bm{q}_{\mathrm{u}, \tilde{t}}\in \mathbb{R}^3$, with $\tilde{t}=1,2,\dots, T/2$, randomly and uniformly within the sphere $\mathcal{S}(\bm{q}_{\mathrm{u}},\delta)$. Then, we design the RIS phase profiles as:
  \begin{align}\label{eq:directionalPhase}
      \tilde{\bm{\omega}}_{\tilde{t}} = \bm{b}(\bm{q}_{\mathrm{u},\tilde{t}}-\bm{p}_{\mathrm{r}})^*.
  \end{align}
It can be seen that by assigning the RIS phase profile according to the latter expression, the reflected signal power is directed towards the points $\bm{q}_{\mathrm{u},\tilde{t}}$.

\section{Fisher Information Analysis}
In this section, we present an analytical lower bound on the estimation error based on the \ac{fim} analysis. For all unbiased estimators, the estimation error is lower-bounded by the \ac{peb}, that is
\begin{align}\label{eq:peb}
    \mathrm{PEB} \triangleq \sqrt{\mathrm{tr}\left([\bm{J}(\bm{\eta})^{-1}]_{3:5,3:5}\right)}.
\end{align}
Here, $\bm{\eta}\in\mathbb{R}^{5+3L}$ is the vector of unknowns that is 
\begin{align}
    \bm{\eta}\triangleq[\rho_0, \varphi_0, \bm{p}_{\mathrm{u}}^\top, \rho_1, \varphi_1, \tau_1,\dots,\rho_L, \varphi_L, \tau_L]^\top,
\end{align}
 where $\rho_l = |\beta_l|$ and $\varphi_l=\angle\beta_l$ for $l=0,1,\dots,L$. Furthermore, $\bm{J}(\bm{\eta})$ is the \ac{fim}, which is defined as \cite{kay1993fundamentals}
\begin{align}\label{eq:FIM}
    \bm{J}(\bm{\eta})\triangleq\frac{2\Es}{\sigma_{\mathrm{n}}^2}\Re \left\{\mathlarger{\mathlarger{\sum}}_{t=1}^T\left(\frac{\partial \bm{\mu}_t}{\partial\bm{\eta}}\right)^{\text{H}} \frac{\partial\bm{\mu}_t}{\partial\bm{\eta}}\right\},
\end{align}
where $\bm{\mu}_t\triangleq\beta_0 \bm{d}(\tau_0) \bm{b}(\bm{p}_{\mathrm{ur}})^{\top} \bm{\omega}_t+ \sum_{l=1}^L\beta_l \bm{d}(\tau_l)$. Next, we study the structure of the \ac{fim} when the temporal sum of the utilized RIS phase profiles (i.e., $\sum_t\bm{\omega}_t$) becomes the all-zero vector. This includes (but is not restricted to) the codebooks discussed in Section~\ref{sec:risPhaseDesign}.  
\begin{proposition}\label{prob:prob1}
In the FIM expression given by \eqref{eq:FIM}, any element of the form:
\begin{align}
    \sum_{t=1}^T\left(\frac{\partial \bm{\mu}_t}{\partial{\eta}_{\text{LOS}}}\right)^{\text{H}} \frac{\partial\bm{\mu}_t}{\partial{\eta}_{\text{NLOS}}}
\end{align}
for any $\eta_{\text{LOS}} \in \bm{\eta}_{\text{LOS}}\triangleq[ \rho_0,\varphi_0,\bm{p}_{\mathrm{u}}^\top]^\top$ and any 
${\eta}_{\text{NLOS}} \in [\rho_1, \varphi_1, \tau_1,\dots,\rho_L, \varphi_L, \tau_L]^\top$ evaluates to zero, provided that  $\sum_t \bm{\omega}_t=\bm{0}$ holds.
\end{proposition}
\begin{proof}
Any term of the form ${\partial \bm{\mu}_t}/{\partial{\eta}_{\text{LOS}}}$ depends linearly on $\bm{\omega}_t$, while any term ${\partial\bm{\mu}_t}/{\partial{\eta}_{\text{NLOS}}}$ is independent of $\bm{\omega}_t$. Since $\sum_t \bm{\omega}_t=\bm{0}$, the corresponding entry in the FIM will evaluate to zero.
\end{proof}

Based on Proposition\,\ref{prob:prob1}, the FIM has the following structure:
\begin{align}
    \bm{J}(\bm{\eta})=\begin{bmatrix}
\bm{J}(\bm{\eta}_{\text{LOS}}) & \bm{0} \\
\bm{0} &  \bm{J}(\bm{\eta}_{\text{NLOS}})
\end{bmatrix},
\end{align}
where $\bm{0}$ indicates the all-zero matrix and $\bm{J}(\bm{\eta}_{\text{LOS}})$ and $\bm{J}(\bm{\eta}_{\text{NLOS}})$ are defined similar to \eqref{eq:FIM}. Using the Schur's complement method, we have that \begin{align}
    [\bm{J}(\bm{\eta})^{-1}]_{1:5,1:5} = \bm{J}(\bm{\eta}_{\text{LOS}})^{-1}.
\end{align}
Hence, we can limit our discussion to the LOS path to be able to calculate the the \ac{peb} in \eqref{eq:peb}. 
The partial derivatives needed in (\ref{eq:FIM}) can be expressed as
\begin{align}
\frac{\partial\bm{\mu}_t}{\partial\rho_0}&=e^{j\varphi_0} \bm{b}(\bm{p}_{\mathrm{ur}})^\top\bm{\omega}_t\, \bm{d}(\tau_0),\\
\frac{\partial\bm{\mu}_t}{\partial\varphi_0}&=j\beta_0 \bm{b}(\bm{p}_{\mathrm{ur}})^\top\bm{\omega}_t\, \bm{d}(\tau_0),\\
\frac{\partial\bm{\mu}_t}{\partial\bm{p}_{\mathrm{u}}}&=\frac{2\beta_0 \bm{b}(\bm{p}_{\mathrm{ur}})^\top\bm{\omega}_t}{c}\dot{\bm{d}}(\tau_0)\,\bm{u}_{\mathrm{ur}}^\top  + \beta_0 \bm{d}(\tau_0) \bm{\omega}^\top_t \dot{\bm{B}}(\bm{p}_{\mathrm{ur}}),
\end{align}
where we have used the following definitions:
\begin{align}
    \dot{\bm{d}}(\tau_0)&\triangleq\frac{\partial\bm{d}(\tau_0)}{\partial\tau}=-j2\pi\Delta f\, \bm{n}\odot\bm{d}(\tau_0), \\
     \dot{\bm{B}}(\bm{p}_{\mathrm{ur}}) & \triangleq \frac{\partial\bm{b}(\bm{p}_{\mathrm{u}})}{\partial\bm{p}_{\mathrm{u}}}\\
      & =-j\frac{4\pi}{\lambda}\left(\text{diag}\left(\bm{b}(\bm{p}_{\mathrm{ur}})\right)\bm{K}^\top-\bm{b}(\bm{p}_{\mathrm{ur}})\bm{u}_{\mathrm{ur}}^\top\right)
\end{align}
with $\bm{u}_{\mathrm{ur}}\triangleq\bm{p}_{\mathrm{ur}}/\Vert \bm{p}_{\mathrm{ur}}\Vert$, $\bm{K}\triangleq[\bm{u}_0, \bm{u}_1, \dots, \bm{u}_{M-1}]$, $\bm{u}_m\triangleq(\bm{p}_{\mathrm{u}}-\bm{p}_{\mathrm{r},m})/\Vert \bm{p}_{\mathrm{u}}-\bm{p}_{\mathrm{r},m} \Vert$, and $\bm{n}\triangleq[0,1,\dots,N-1]$. 
Putting all of the above together, the FIM for the LOS path  is
\begin{align}
  \bm{J}(\bm{\eta}_{\text{LOS}})&= \frac{2\Es}{\sigma_{\mathrm{n}}^2} \begin{bmatrix}
NG & 0 & \Re\{\bm{z}^{\top}\}\\
0 & N|\beta_{0}|^{2}G & |\beta_{0}|\Im\{\bm{z}^{\top}\}\\
\Re\{\bm{z}\} & |\beta_{0}|\Im\{\bm{z}\} & \bm{H}
\end{bmatrix},
\end{align}
where we have used the definitions:
\begin{align}
G & \triangleq\sum_{t}\vert\bm{b}(\bm{p}_{\mathrm{ur}})^{\top}\bm{\omega}_{t}\vert^{2},\\
\end{align}
\begin{align}
\bm{z}^{\top}&\triangleq \frac{2 |\beta_{0}| }{c} G \bm{d}^\mathrm{H}\dot{\bm{d}} \bm{u}_{\mathrm{ur}}^{\top}  + N|\beta_{0}| \bm{b}^{H}\mathbf{C}\dot{\bm{B}}, \\
\bm{H}&\triangleq \frac{4|\beta_{0}|^{2}}{c^{2}}G\Vert\dot{\mathbf{d}}\Vert^{2}\mathbf{u}_{\mathrm{ur}}\mathbf{u}_{\mathrm{ur}}^{\top}\nonumber\\
&\quad+|\beta_{0}|^{2}N\Re\left\{ \dot{\mathbf{B}}^{H}\mathbf{C}\dot{\mathbf{B}}\right\} +\frac{2|\beta_{0}|^{2}}{c}\Re\left\{\bm{U}+\bm{U}^{\mathrm{H}}\right\}.
\end{align}
In the latter expressions, we have dropped the dependencies on $\tau_0$ and $\bm{p}_{\mathrm{ur}}$ for notation simplification. Furthermore, we used:
\begin{align}
    \mathbf{C} & \triangleq\sum_{t}(\bm{\omega}_{t}^{*}\bm{\omega}_{t}^{\top}),\\
\bm{U}&\triangleq\left(\dot{\mathbf{d}}^{H}\mathbf{d}\right)\mathbf{u}_{\mathrm{ur}}\bm{b}^{H}\mathbf{C}\dot{\mathbf{B}}.
\end{align}

By using Schur's complement, we can calculate the equivalent FIM (EFIM) of $\bm{p}_{u}$ as
\begin{align} \label{eq:EFIM}
    \bm{J}&(\bm{p}_{\mathrm{u}})= \frac{8|\beta_0|^2\Es G}{c^2\sigma_{\mathrm{n}}^2}\left(\Vert\dot{\bm{d}}\Vert^2-\frac{\vert \dot{\bm{d}}^\mathrm{H}\bm{d}\vert^2}{N}\right)\mathbf{u}_{\mathrm{ur}}\mathbf{u}_{\mathrm{ur}}^{\top} \\
    &+\frac{2\Es |\beta_0|^2 N}{\sigma_{\mathrm{n}}^2}\left(\Re\{\dot{\bm{B}}^{\mathrm{H}} F \dot{\bm{B}}\} -\frac{1}{G}\Re\{\dot{\bm{B}}^{\mathrm{H}} \bm{F} \bm{b}\bm{b}^{\mathrm{H}}\bm{F}\dot{\bm{B}}\}\right).\nonumber
\end{align}
Then, the PEB in \eqref{eq:peb} can be computed as
\begin{align}
    \mathrm{PEB} = \sqrt{\mathrm{tr}\left(\bm{J}(\bm{p}_{\mathrm{u}})^{-1}\right)}.
\end{align}

\begin{rem}\label{remark:peb}
In far-field conditions, the vector $\bm{b}(\bm{p}_{\mathrm{ur}})$ can be approximated by $\bm{b}_{\mathrm{ff}}$, having the elements:
\begin{align}\label{eq:bff}
[\bm{b}_{\mathrm{ff}}(\bm{p}_{\mathrm{ur}})]_m=\exp \left(\jmath \frac{4\pi}{\lambda} \bm{u}_{\mathrm{ur}}^\top (\bm{p}_{\mathrm{r},m}-\bm{p}_r) \right).
\end{align}
In addition, the matrix $\dot{\bm{B}}(\bm{p}_{\mathrm{ur}})$ is approximated by 
\begin{align}
    \dot{\bm{B}}_{\mathrm{ff}}(\bm{p}_{\mathrm{ur}})&=\frac{\partial \bm{b}_{\mathrm{ff}}(\bm{p}_{\mathrm{ur}})}{\partial \bm{p}_{\mathrm{ur}}}\\
    & =\jmath\frac{4\pi}{\lambda\Vert\bm{p}_{\mathrm{ur}}\Vert}\mathrm{diag}(\bm{b}_{\mathrm{ff}}(\bm{p}_{\mathrm{ur}}))\bm{M}^\top\left(\bm{I}-\bm{u}_{\mathrm{ur}}\bm{u}_{\mathrm{ur}}^\top\right),\nonumber
\end{align}
where the matrix $\bm{M}$ is given by
\begin{align}
\bm{M}=\left[\bm{p}_{\mathrm{r},0}-\bm{p}_{\mathrm{r}},\bm{p}_{\mathrm{r},1}-\bm{p}_{\mathrm{r}},\dots,\bm{p}_{\mathrm{r},M-1}-\bm{p}_{\mathrm{r}}\right],
\end{align}
and contains the relative positions of the RIS elements with respect to the surface's center. This implies that the EFIM in \eqref{eq:EFIM} consists of two components. The first one is proportional to $\bm{u}_{\mathrm{ur}}\bm{u}_{\mathrm{ur}}^\top$, and it is scaled by the signal bandwidth (due to the term $\Vert\dot{\bm{d}}(\tau_0)\Vert$) and the energy reflected by the RIS towards the UE. The second component is proportional to $\mathbf{I}-\bm{u}_{\mathrm{ur}}\bm{u}_{\mathrm{ur}}^\top$, which means that it is in the two-dimensional subspace orthogonal to the first component. Moreover, this second component decreases with the UE distance as $1/\Vert\bm{p}_{\mathrm{ur}}\Vert^2$ (without accounting here for any possible dependence of $\beta_0$ with the distance), and it is independent of the bandwidth. Note that the first component represents the contribution of the estimation of the propagation delay to the EFIM, while the second component corresponds to the contribution of the estimation of the direction of the UE (i.e., the angles of arrival and departure on/from the RIS). We next consider the case where the $\bm{u}_{\mathrm{ur}}$ is orthogonal to the RIS normal $\bm{n}_{\mathrm{r}}$. Note that since $\bm{M}^\top \bm{n}_{\mathrm{r}} = 0$, we have that  \begin{align}
    \bm{M}^\top\left(\mathbf{I}-\bm{u}_{\mathrm{ur}}\bm{u}_{\mathrm{ur}}^\top\right) = \bm{M}^\top\left(\mathbf{I}-\bm{u}_{\mathrm{ur}}\bm{u}_{\mathrm{ur}}^\top-\bm{n}_{\mathrm{r}}\bm{n}_{\mathrm{r}}^\top\right).\label{eq:nnt}
\end{align}
Therefore, if $\bm{u}_{\mathrm{ur}}^\top \bm{n}_{\mathrm{r}}=0$ holds, the term $\mathbf{I}-\bm{u}_{\mathrm{ur}}\bm{u}_{\mathrm{ur}}^\top-\bm{n}_{\mathrm{r}}\bm{n}_{\mathrm{r}}^\top$ becomes a one-dimensional subspace, which yields an ill-conditioned FIM. 
\end{rem}

\section{Low-complexity Localization}\label{sec:LowCompLocalization}
In this section, we present a low-complexity localization method for the considered RIS-aided wireless system. The proposed estimator has three steps: \emph{i)} we first obtain a coarse estimate of the delay $\tau_0$, then \emph{ii)} we compute a coarse estimate of the relative UE position $\bm{p}_{\mathrm{ur}}$, and finally \emph{iii)} we refine our estimation via maximizing the likelihood function.

\subsection{Coarse Estimation of $\tau_0$}
The delay can be estimated by performing the inverse FFT (IFFT) of each vector $\tilde{\bm{y}}_t$ and non-coherently accumulating the results. Specifically, we calculate the quantity:
\begin{align}
    \bm{y}_{\mathrm{f}}\triangleq\sum_{t=1}^{T/2}\Vert \bm{F}\tilde{\bm{y}}_t \Vert^2,
\end{align}
where $\bm{F}$ is the $N^{\prime} \times N$ IFFT matrix with elements
$
[\bm{F}]_{r,s} = \exp\left(\jmath 2\pi r s/N'\right)/\sqrt{N'}
$, where $N'$ is a design parameter. 
Next, we estimate the delay as
\begin{align}\label{eq:tauHat}
    \hat{\tau}_0 = \frac{1}{N'\Delta f}\arg\max_n\vert[\bm{y}_{\mathrm{f}}]_n\vert.
\end{align}

\subsection{Coarse Estimation of $\mathbf{p}_{\mathrm{ur}}$}\label{sec:coarsePest}
Once the estimation $\hat{\tau}_0$ is available, we can generate the following sequence:
\begin{align}
    \bm{z}\triangleq\bm{d}(\hat{\tau}_0)^{\text{H}}\left[\tilde{\bm{y}}_1, \dots, \tilde{\bm{y}}_{T/2}\right],
\end{align} 
and formulate the cost function  $P(\bm{p})\triangleq{|\bm{s}(\bm{p})\bm{z}^\text{H}|^2}/{\Vert \bm{s}(\bm{p})\Vert^2},$
where
$    \bm{s}(\bm{p})
    =\bm{b}(\bm{p})^{\top}\left[\tilde{\bm{\omega}}_1, \tilde{\bm{\omega}}_2, \dots, \tilde{\bm{\omega}}_{T/2}\right].$%
Using this cost function, we can estimate the UE position as
$\hat{\bm{p}}_{\mathrm{ur}} \triangleq \arg\max_{\bm{p}} P(\bm{p}).    $%
The evaluation of $P(\bm{p})$ is relatively simple since the vectors $\bm{s}(\bm{p})$ can be precomputed on a grid of points, which can be stored at the RIS controller and made available to the UE. Moreover, $P(\bm{p})$ does not need to be evaluated in the full three-dimensional (3D) grid of points, but only in a 2D subgrid of points that fulfill $c\hat{\tau}_0/2 - \epsilon \le \Vert \bm{p}\Vert \le c\hat{\tau}_0/2 + \epsilon $, where $\epsilon$ accounts for the range of errors in $\hat{\tau}_0$ and the density of the position grid. 

\subsection{Maximum Likelihood Position Estimation}
We use the estimated $\hat{\bm{p}}_{\mathrm{ur}}$ in Sec.\,\ref{sec:coarsePest} as a starting point to optimize the \ac{ml} function via a quasi-Newton algorithm.
Based on \eqref{eq:yTilde}, the \ac{ml} estimator can be written as
\begin{align}
    \hat{\bm{p}}_{\mathrm{ur}} &\triangleq \arg\min_{\bm{p}_{\mathrm{ur}}}\min_{\beta_0} \sum_{t=1}^{T/2}\Vert \beta_0 \bm{\zeta}_t(\bm{p}_{\mathrm{ur}})-\tilde{\bm{y}}_t\Vert^2\\
    &= \arg\min_{\bm{p}_{\mathrm{ur}}} \sum_{t=1}^{T/2}\Vert \beta_0(\bm{p}_{\mathrm{ur}}) \bm{\zeta}_t(\bm{p}_{\mathrm{ur}})-\tilde{\bm{y}}_t\Vert^2\label{eq:mlfunction},
\end{align}
where 
$    \bm{\zeta}_t(\bm{p}_{\mathrm{ur}})\triangleq \bm{d}\left({2\Vert\bm{p}_{\mathrm{ur}}\Vert}/{c}\right) \bm{b}(\bm{p}_{\mathrm{ur}})^\top\tilde{\bm{\omega}}_t,$ and $
    \beta_0(\bm{p}_{\mathrm{ur}}) \triangleq \sum_t{\bm{\zeta}_t(\bm{p}_{\mathrm{ur}})^{\mathrm{H}}\tilde{\bm{y}}_t}/({\bm{\zeta}_t(\bm{p}_{\mathrm{ur}})^{\mathrm{H}}\bm{\zeta}_t(\bm{p}_{\mathrm{ur}})}).$
Finally, the UE position can be estimated as $\hat{\bm{p}}_{\mathrm{u}}=\hat{\bm{p}}_{\mathrm{ur}}+{\bm{p}}_{\mathrm{r}}$.

\section{Numerical Results}\label{sec:simulationResults}

\begin{table}[!t]
\vspace{.1cm}
	\caption{Parameters Used in the Simulation Results.}
	\label{table:par}
	\centering
	\begin{tabular}{l l l }
		\hline
		\hline
		Parameter&Symbol& Value\\
		\hline
		Carrier frequency & $f_{\mathrm{c}}$ & $28 \ {\mathrm{GHz}}$\\
		Number of RIS elements & $M$ &$100\times 100$\\
		Light speed & $c$ & $3\times 10^8 \ \mathrm{m/s}$\\
		RIS inter-element distance & $d$ &$ \lambda/4$\\
		Number of subcarriers & $N$ & $3\, 000$\\
		Subcarrier bandwidth & $\Delta F$ & $120 \ \mathrm{kHz}$\\
		Number of transmissions & $T$ & $100$\\
		Transmit power &$N\Delta F E_{\mathrm{s}}$ & $23 \ \mathrm{dBm}$\\
		Noise power spectral density & $N_0$ & $-174 \ \mathrm{dBm/Hz}$\\
		UE's noise figure& $n_{\mathrm{f}}$ & $3 \ \mathrm{dB}$\\
		Noise variance& $\sigma^2_{\mathrm{n}}\triangleq n_{\mathrm{f}} N_0$ & $-171 \ \mathrm{dBm/Hz}$\\
		\hline
		\hline
	\end{tabular}
	\vspace{-.4cm}
\end{table}

\begin{figure*}
    \centering
    \input{Results}
    \caption{Positioning errors for a system with a $100\times 100$ RIS at the origin. (a) Estimation error (markers) and PEBs (solid lines) for a UE located at $\bm{p}_{\mathrm{u}}=(d,d,d)/\sqrt(3)$, with random and directional RIS phase designs, (b) \ac{cdf} of the estimation error (dashed lines) and PEB (solid lines) for  $d=11 \mathrm{m}$ and $100$ realizations of RIS phase profiles, (c) PEB at $\bm{p}_{\mathrm{u}}=[10,10,10]/\sqrt{3}$ for different RIS sizes ($M$) and for directional and random codebooks.}
    \label{fig:results}
    \vspace{-.4cm}
\end{figure*}
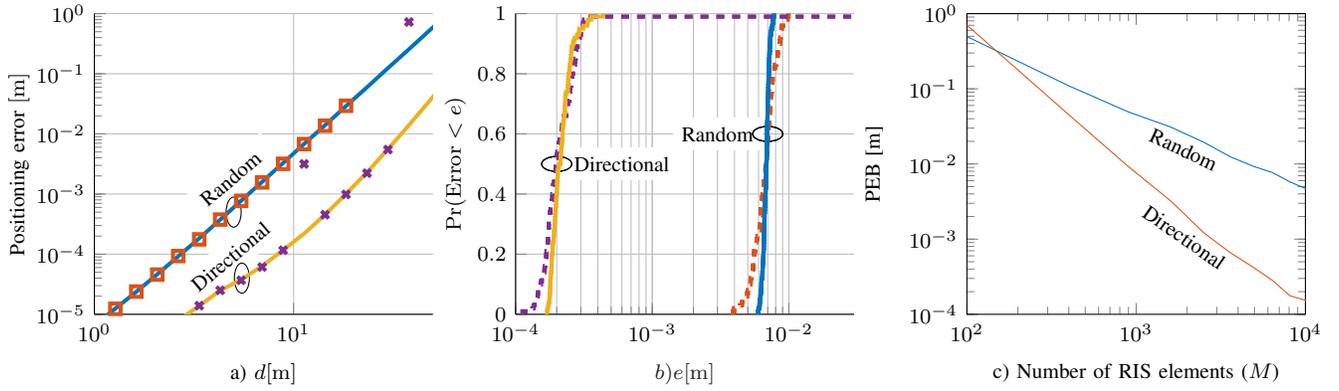

\begin{figure}[!t]
    \centering
    \input{FigColorMaps}
    \caption{Contour plot of the PEB in meters for the UE position $[x,y,y]$ for a random RIS phase profile, when the path loss is set a) according to \eqref{eq:beta}, and b) equally for all the points (to the path loss at $[10,10,10]/\sqrt(3)$).  }
    \label{fig:M}
    \vspace{-.4cm}
\end{figure}
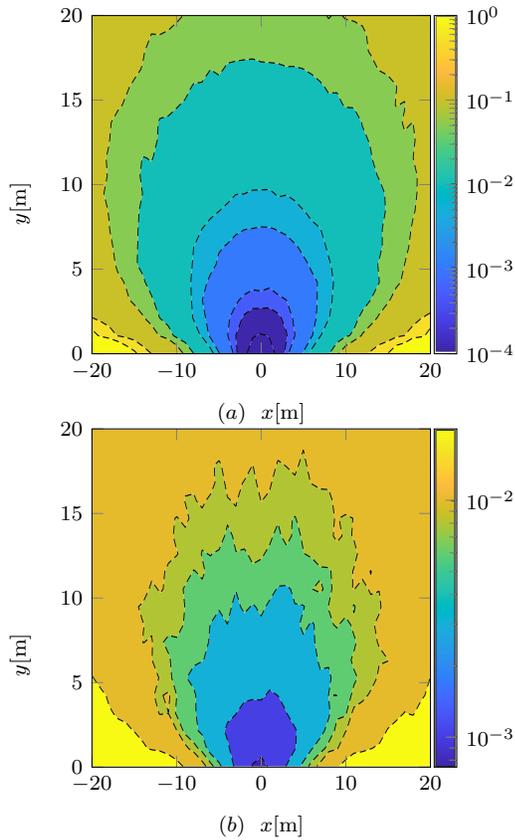

In this section, we measure the accuracy of the proposed estimation method and calculate the PEBs for an example system whose parameters are listed in Table~\ref{table:par}, where $N$ and $\Delta F$ are chosen based on the 5G NR numerology. Furthermore, we set $N'=10 N$. We assume that the RIS is located at the origin of the coordinate system with elements in the $z=0$ plane. The channel gain is calculated as \cite[Eq. (3)]{tang2020wireless}
\begin{align}\label{eq:beta}
    \beta_0 = \frac{\lambda^2 \cos(\phi)}{16\pi^{1.5} \Vert\bm{p}_{\mathrm{ur}}\Vert^2},
\end{align}
where $\phi =\cos^{-1}(\bm{u}_{\mathrm{ur}}^\top\bm{n}_{\mathrm{r}})$ indicates the angle between the RIS normal ($\bm{n}_{\mathrm{r}}$) and $\bm{p}_{\mathrm{ur}}$. The distance between adjacent RIS elements was set to $d=\lambda/4$ based on Remark\,\ref{remark:lambda}.

In Fig.\,\ref{fig:results}\,(a), we compare the positioning error of the proposed estimation method in Section\,\ref{sec:LowCompLocalization} with the PEB in \eqref{eq:peb}. The UE is located at  $[d,d,d]/\sqrt{3}$ meters, where $d\in[1,50]$. We consider the two RIS phase profile designs mentioned in Section\,\ref{sec:risPhaseDesign}, specifically, the random and directional with $\delta=1\,\mathrm{m}$. For each UE location the results are the average over $100$ realizations of RIS phase profiles, for each of which $10$ noise realizations were generated to calculate the positioning error. It can be seen that the proposed estimator can attain the theoretical lower bound for up to $d=18$\,m for the random phase design and $d=29$\,m for the directional one (with one exception point, which will be discussed later). With higher values of $d$, we can see from the PEBs that sub-meter UE localization is still possible, however, our low-complexity localization cannot achieve the bounds due to the low \ac{snr}.

In Fig.\,\ref{fig:results}\,(a), for directional profiles, there exists an exception point, which is at $d=11$\,m, where the estimation error is considerably higher than the PEB. This is due to the existence of one outlier. As mentioned in Section\,\ref{sec:directionalCodeBook}, the directional codebook is generated by directing the beam towards random points inside the uncertainty area. If none of these points are in the close vicinity of the UE, then the received \ac{snr} is low and the estimator fails to localize. The probability of this event depends on the uncertainty area and the RIS beamwidth. Nonetheless, for the parameters considered in this paper, we can see that it is low. Figure\,\ref{fig:results}\,(b) shows the cumulative distribution function (CDF) for the PEB (solid lines) and estimation errors (dashed lines) for the point at  $d=11$\,m. It can be seen that the estimation error follows the PEB closely. Furthermore, it is shown that, due to the aforementioned outlier, the \ac{cdf} of the estimation error for the directional profiles is saturated close to $1$ (at $0.99$).

In Fig.\,\ref{fig:results}\,(c), we illustrate the effects of the RIS size on the PEB at the UE location $\bm{p}_{\mathrm{u}}=[10,10,10]/\sqrt{3}$. It can be seen that, for the directional phase profile, the PEB reduces with $M$ faster than for the the random one. This is due to the fact that the SNR grows faster with $M$ when directional beamforming is used. Furthermore, with very low values of $M$, different beams in directional codebook becomes almost identical (due to the very large beamwidths), which reduces the estimation accuracy. Hence, the directional codebook performs worse than the random one at (very) low values of $M$.

In Fig.\,\ref{fig:M}\,(a), we show the PEB for the \ac{ue} position at $[x,y,y]$, where $x\in[-20,20]$ and $y\in[0,20]$. We consider one realization of  random RIS phases. It can be seen that the PEB is almost symmetrical with respect to the $y$ axis. Lower PEBs are obtained for lower values of $\theta$, (partly) because of the RIS unit-cell directional pattern, which is incorporated in \eqref{eq:beta}. As demonstrated, submeter localization is possible almost all across the considered area. In Fig.\,\ref{fig:M}\,(b), we study the PEB by setting $\beta_0$ to be the same for all the UE positions. They are set to the $\beta_0$ calculated at $[10,10,10]/\sqrt{3}$. By doing so, we can remove the effects of the SNR on PEB and focus on the geometrical effects. It can be seen that PEB increases with the UE distance and also $\theta$; this behavior was explained in Remark\,\ref{remark:peb}.

\if{0}
\begin{figure}
    \centering
    \includegraphics[width=0.9\columnwidth]{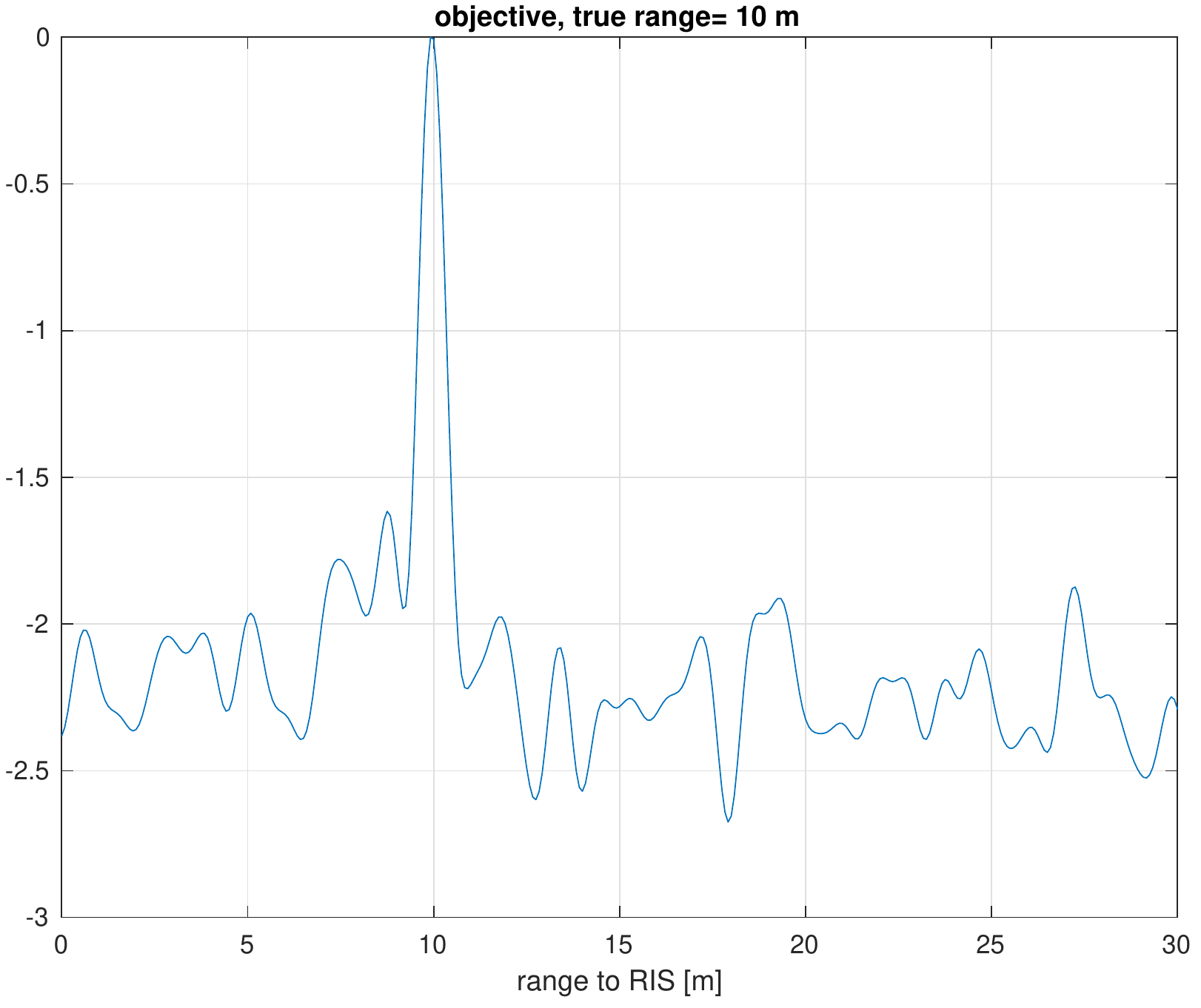}
    \caption{Example of range estimation.}
    \label{fig:rangeEstimate}
\end{figure}

\begin{figure}
    \centering
    \includegraphics[width=0.9\columnwidth]{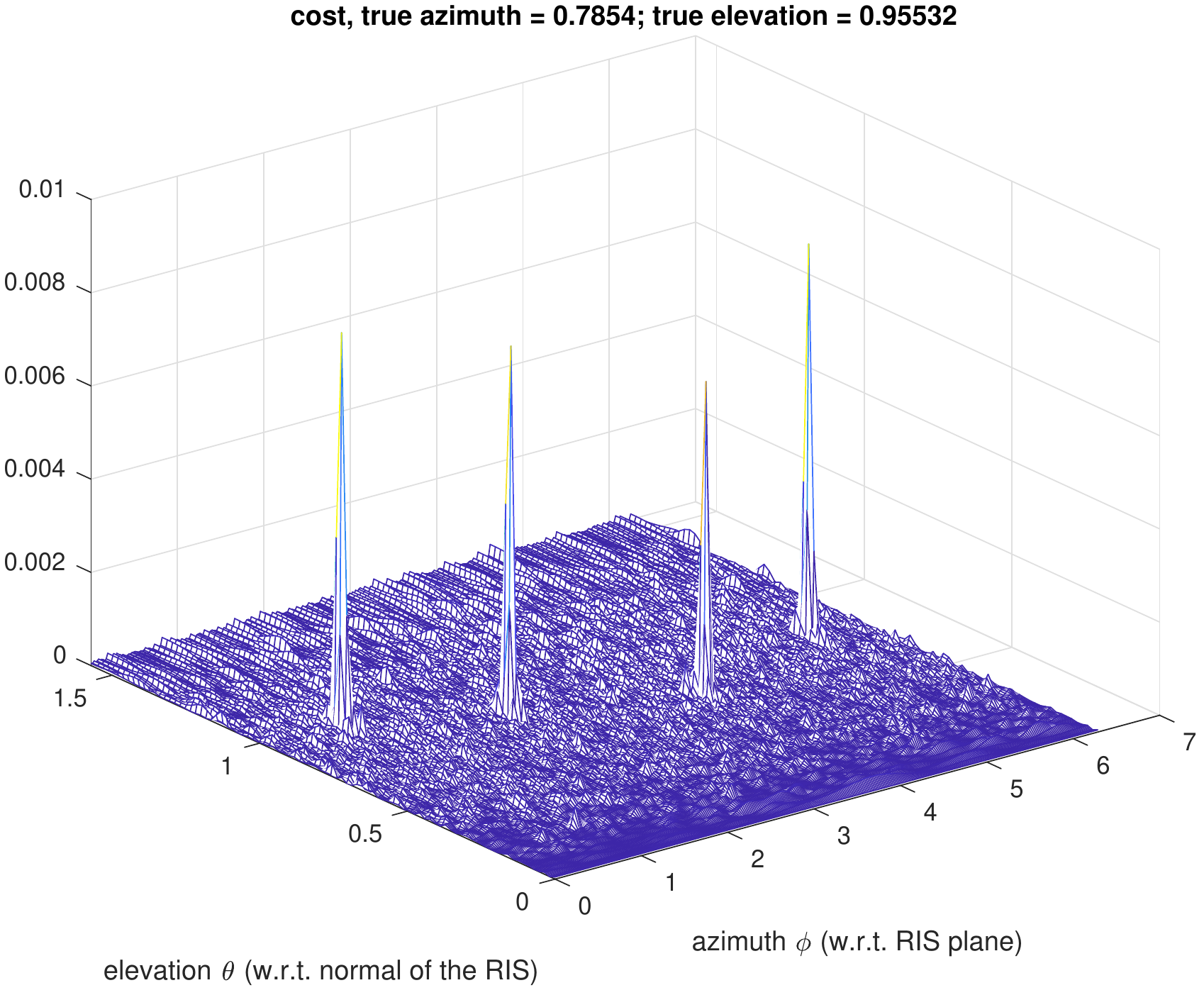}
    \caption{Example of angle estimation.}
    \label{fig:angleEstimate}
\end{figure}
\fi

\vspace{-.4cm}
\section{Conclusion}
We presented the concept of 3D UE self-localization with only one RIS, where the UE transmits multiple OFDM signals and processes the reflected signal from the RIS. The positioning was performed by \emph{i)} separating the signal reflected from the RIS from the undesired multipath, \emph{ii)} establishing a coarse estimate of the position, and \emph{iii)} refining the estimation via ML. We assessed the performance of the estimator in terms of the positioning error and compared it to an analytical lower bound. Our results provides a novel example where RISs become an enabling technology for radio localization. Note that, if in the proposed scenario, the RIS is replaced by a generic scatterer with a known location, then we could only estimate the distance $\Vert\bm{p}_{\mathrm{ur}}\Vert$ (even in the absence of multipath). This of course would not be enough to find out the location of the UE.  Future work includes the investigation of the effects of RIS impairments, NLOS components, and the UE mobility
on the proposed system model and estimation approach.


\vspace{.5cm}{\footnotesize
\paragraph*{\footnotesize Acknowledgments}
This work was supported, in part, by the Swedish Research Council under grant 2018-03701, the EU H2020 RISE-6G project under grant 101017011, the Spanish Ministry of Science and Innovation PID2020-118984GB-I00 and by the Catalan ICREA Academia Programme.}

\if{0}
\newpage
{\color{gray} Do we adopt a far-field model where vector $\bm{a}$ depends only on the direction? Or do we keep the general model and maybe particularize later?}
\George{If the treatment holds only for far-field, but it will be tested in the results for near-field too, it's fine to keep it general here. We should directly start with the far field, however, if the analysis is based solely on this.}

{\color{gray} If we consider it relevant, we could make $\bm{d}(\tau)$ a function of $t$, i.e. $\bm{d}_t(\tau)$, in order to account for the stair-like sequences of 3GPP, like PRS. } \George{it is more general this way, if it can be explicitly treated. Of course, this will depend on the actual environment. Will we focus on a specific one, say indoor and a static user?}

\balance
\bibliography{refs}
\fi


\end{document}

%% file: Results.tex
\definecolor{mycolor1}{rgb}{0.00000,0.44700,0.74100}%
\definecolor{mycolor2}{rgb}{0.85000,0.32500,0.09800}%
\begin{tikzpicture}[every text node part/.style={align=center}]
\pgfplotsset{every tick label/.append style={font=\footnotesize}}
\input{M}
\input{pebs}
\input{cdfs}
\node [rotate = 90] at (2.7cm,2cm){\footnotesize $\mathrm{Pr}(\mathrm{Error}<e)$};
\node [rotate = 90] at (-3.1cm,2cm){\footnotesize Positioning error $[\mathrm{m}]$};
\end{tikzpicture}

%% file: FigColorMaps.tex
\begin{tikzpicture}[every text node part/.style={align=center}]
\pgfplotsset{every tick label/.append style={font=\footnotesize}}
\begin{axis}[
at={(8.95cm,0)},
axis on top,
width=4.5cm,
height=4.5cm,
xlabel={\footnotesize $(a) \ \ x [\mathrm{m}]$},
scale only axis,
xmin=-20,
xmax=20,
ymin=0,
ymax=20,
enlargelimits=false, 
axis on top, 
]
\addplot[thick,blue] graphics[xmin=-20,ymin=0,xmax=20,ymax=20] {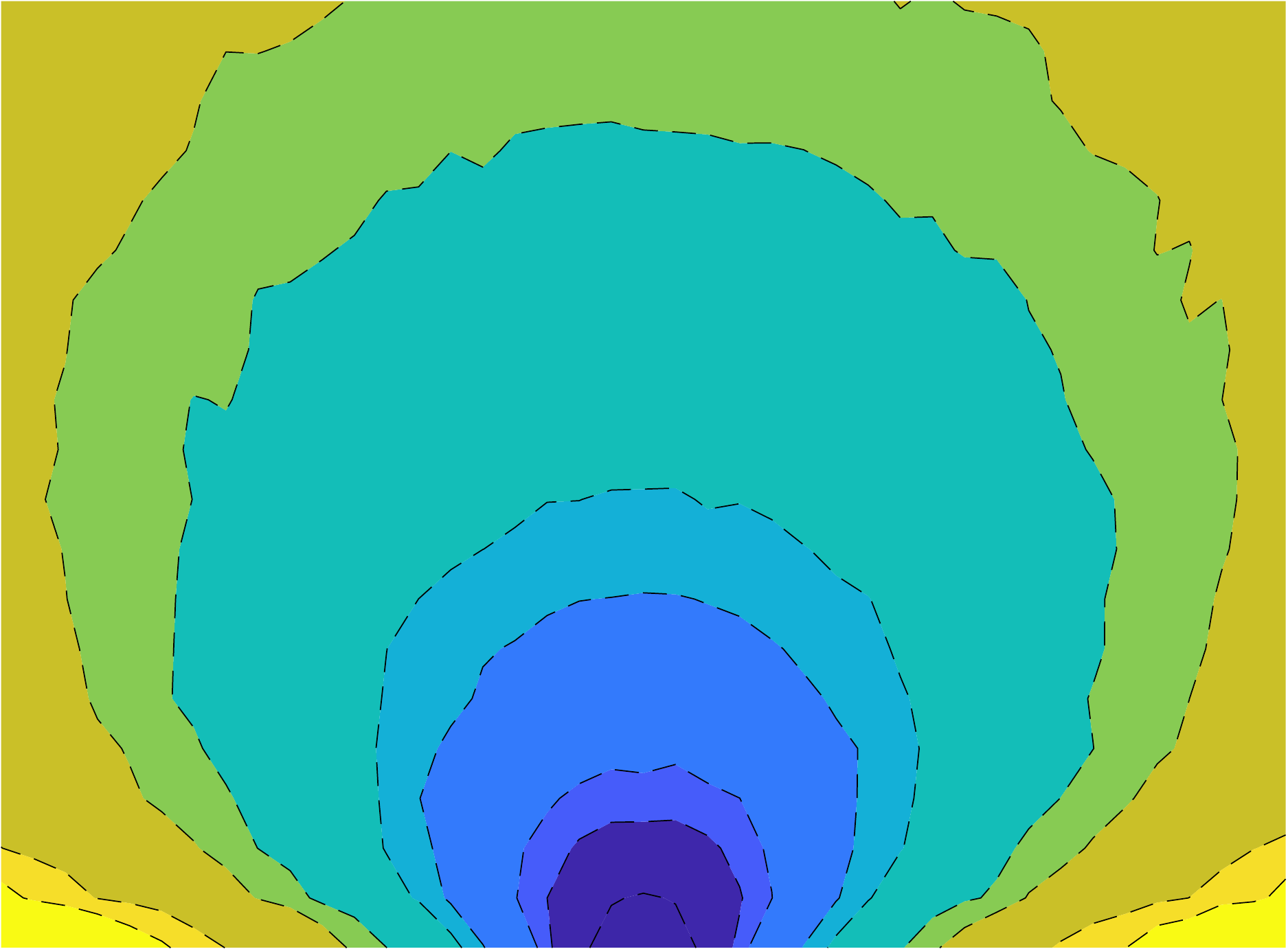};
\end{axis}
\begin{axis}[at={(13.5cm,0cm)},axis on top,
		width=.3cm,
		height=4.5cm,
		scale only axis,
		xmin=0,
		xmax=1,
		ymin=1e-4,
		ymax=1,
		ymode = log,
		xtick=\empty,
		enlargelimits=false, 
		axis on top, 
        ytick pos=right,
]
\addplot[thick,blue] graphics[xmin=0,ymin=1.04e-4,xmax=1,ymax=1] {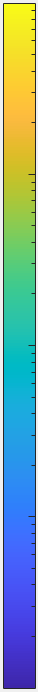};
\end{axis}
\node [rotate = 90] at (8cm,2cm){\footnotesize $y[\mathrm{m}]$};
\begin{axis}[
at={(8.95cm,-5.5cm)},
axis on top,
width=4.5cm,
height=4.5cm,
xlabel={\footnotesize $(b) \ \ x [\mathrm{m}]$},
scale only axis,
xmin=-20,
xmax=20,
ymin=0,
ymax=20,
enlargelimits=false, 
axis on top, 
]
\addplot[thick,blue] graphics[xmin=-20,ymin=0,xmax=20,ymax=20] {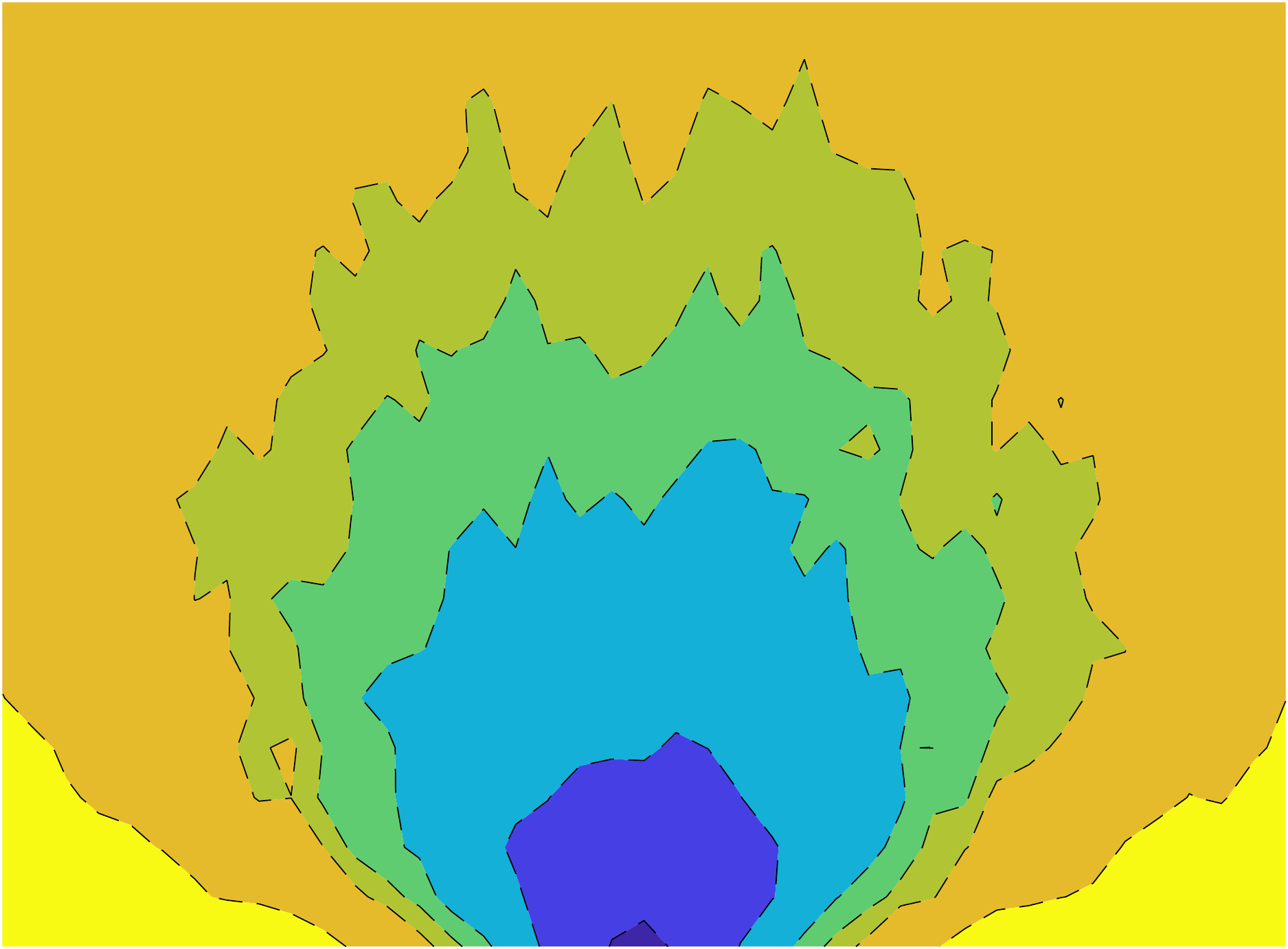};
\end{axis}
\begin{axis}[at={(13.5cm,-5.5cm)},axis on top,
		width=.3cm,
		height=4.5cm,
		scale only axis,
		xmin=0,
		xmax=1,
		ymin=7.47e-4,
		ymax=.02,
		ymode = log,
		xtick=\empty,
		enlargelimits=false, 
		axis on top, 
        ytick pos=right,
]
\addplot[thick,blue] graphics[xmin=0,ymin=7.47e-4,xmax=1,ymax=.02] {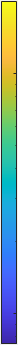};
\end{axis}
\node [rotate = 90] at (8cm,-4cm){\footnotesize $y[\mathrm{m}]$};
\end{tikzpicture}